\documentclass[10pt, doublesided, doublecolumn,final]{IEEEtran}
\usepackage{amssymb, amsmath, color}
\usepackage{amsthm,graphicx,enumerate,verbatim,xcolor}
\usepackage[small]{caption}
\usepackage{subfigure}
\usepackage{bbm}
\usepackage{mathtools}
\usepackage{tikz}
\usetikzlibrary{arrows,positioning,fit,backgrounds}
\usetikzlibrary{calc}
\usepackage[top=1in,bottom=0.75in,right=0.75in,left=0.75in]{geometry}

\newtheorem{thm}{Theorem}
\newtheorem{cor}{Corollary}
\newtheorem{lem}{Lemma}
\newtheorem{defn}{Definition}

\newtheorem{property}{Property}
\theoremstyle{remark}

\tikzstyle{arw}=[->,>=latex]
\tikzstyle{node}=[draw,rectangle,rounded corners]

\tikzstyle{col1}=[fill=red!80!]
\tikzstyle{col2}=[fill=green!80!black]
\tikzstyle{col3}=[fill=blue!80!]

\def\squarebox#1{\hbox to #1{\hfill\vbox to #1{\vfill}}}
\newcommand{\inout}{-}
\newcommand{\Pbar}{\overline{P}}

\newcommand{\eps}{\varepsilon}

\newcommand{\Ebb}{\mathbb{E}}

\newcommand{\Pbb}{\mathbb{P}}
\newcommand{\Rbb}{\mathbb{R}}

\newcommand{\Ccal}{\mathcal{C}}

\newcommand{\Mcal}{\mathcal{M}}

\newcommand{\Ucal}{\mathcal{U}}
\newcommand{\Vcal}{\mathcal{V}}

\newcommand{\Xcal}{\mathcal{X}}
\newcommand{\Ycal}{\mathcal{Y}}

\newcommand{\Lcal}{\mathcal{L}}
\newcommand{\cn}{\mathcal{C}^n}
\newcommand{\Pbf}{\mathbf{P}}
\newcommand{\Qbf}{\mathbf{Q}}

\begin{document}

\sloppy

%% Paper Title
%% You can use linebreaks \\ within to get better formatting as
%% desired. 
\title{A Rate-Distortion Based Secrecy System with Side Information at the Decoders}

%% Author names and affiliations:

\author{Eva C. Song \qquad Paul Cuff \qquad H. Vincent  Poor\\ Dept. of Electrical Eng., Princeton University,  NJ 08544\\ \{csong, cuff, poor\}@princeton.edu}

%% To balance the two columns, you should reduce the text-height of
%% the last page using the following command:
%%%%%%%%%%%%%%%%%%%%%%%%%%%%%%%%%%%%%%%%%%%%%%%%%%%%%%%%%%%%%%%%%%%%%
%\addtolength{\textheight}{-9.35cm}
%%%%%%%%%%%%%%%%%%%%%%%%%%%%%%%%%%%%%%%%%%%%%%%%%%%%%%%%%%%%%%%%%%%%%
%% with an appropriate value. This command must be place on the second
%% last page, i.e., for a one-page abstract here, for a two-page
%% abstract right after the \maketitle command.

%% Create the title:
\maketitle
%\blfootnote{This research was supported in part by the Air Force Office of Scientific Research under Grant FA9550-12-1-0196 and MURI Grant FA9550-09-05086 and in part by National Science Foundation under Grants CCF-1116013, CNS-09-05086 and CCF-1350595.}
\begin{abstract}
A secrecy system with side information at the decoders is studied in the context of lossy source compression over a noiseless broadcast channel. The decoders have access to different side information sequences that are correlated with the source. The fidelity of the communication to the legitimate receiver is measured by a distortion metric, as is traditionally done in the Wyner-Ziv problem. The secrecy performance of the system is also evaluated under a distortion metric. An achievable rate-distortion region is derived for the general case of arbitrarily correlated side information. Exact bounds are obtained for several special cases in which the side information satisfies certain constraints. An example is considered in which the side information sequences come from a binary erasure channel and a binary symmetric channel.
\end{abstract}

\section{Introduction}
The wire-tap channel with side information at the decoders has been previously investigated. It was studied in \cite{villard-allerton} under an equivocation constraint at the eavesdropper and a complete characterization of the rate-distortion-equivocation region was derived. A related problem with coded side information was studied in \cite{gunduz}. However, using equivocation as the description of secrecy does not capture how much distortion will occur if the eavesdropper is forced to reconstruct the source. In this work, both the legitimate receiver and the eavesdropper's reconstructions of the source are measured by distortion. Furthermore, the eavesdropper is assumed to make the best use of her side information along with the encoded message. This setting can also be interpreted as a game-theoretic model where the two receivers are playing a zero-sum game and each one is required to output a sequence that is closest to the source sequence being transmitted.

This distortion-based notion of secrecy was also used in \cite{cuff-allerton2010}, \cite{schieler-isit2012} and \cite{song-isit2013} with the presence of secret key sharing between the encoder and the legitimate receiver. It was shown in \cite{schieler-isit2012} that a secret key with any strictly positive rate can force the eavesdropper's reconstruction of the source to be as bad as if she knows only the source distribution, i.e. the distortion under perfect secrecy. This result suggests, if instead of a shared secret key, the decoders have access to different side information, we should be able to force the eavesdropper's reconstruction of the source to be the distortion under perfect secrecy as long as the legitimate receiver's side information is somewhat stronger than the eavesdropper's side information with respect to the source. This is indeed the case, which will be formally stated herein. However, in the more general case, the legitimate receiver may not have the stronger side information. Can a positive distortion still be forced upon the eavesdropper? We will show in this paper that we can encode the source in favor of the legitimate receiver's side information so that the eavesdropper can only make limited use of the encoded message even with the help of her side information.

The proof technique used in the achievability in this paper follows the same line as \cite{cuff-itw2013}, \cite{song-isit2014}, which relies on the soft-covering lemmas. This approach differs from the traditional joint-typicality and random-binning based proofs in that it requires no effort on bounding the different kinds of error events, and the results obtained from this approach apply to both discrete and continuous alphabets, since the soft-covering lemmas have no restriction on the alphabet size.

\section{Preliminaries} \label{prelim}
\subsection{Notation} \label{notation}
A sequence $X_1,..., X_n$ is denoted by $X^n$. Limits taken with respect to ``$n\rightarrow \infty$" are abbreviated as ``$\rightarrow_n$". Inequalities with $\limsup_{n\rightarrow \infty}h_n\leq h$ and $\liminf_{n\rightarrow \infty}h_n\geq h$ are abbreviated as $h_n\leq_n h$ and $h_n\geq_n h$, respectively. %Throughout this paper, random variables are assumed to take only countably many values. 
When $X$ denotes a random variable, $x$ is used to denote a realization, $\mathcal{X}$ is used to denote the support of that random variable, and $\Delta_{\Xcal}$ is used to denote the probability simplex of distributions with alphabet $\Xcal$. The symbol $|\cdot|$ is used to denote the cardinality. A Markov relation is denoted by the symbol $\inout$. %For an i.i.d. sequence pair $(X^n,Y^n)$ distributed according to $\prod_{i=1}^n \Pbar_{XY}$, the joint distribution is denoted by $\Pbar_{X^nY^n}$ and the corresponding conditional probability $\prod_{i=1}^n\Pbar_{Y|X}$ is denoted by $\Pbar_{Y^n|X^n}$. 
We use $\Ebb_P$, $\Pbb_P$, and $I_{P}(X;Y)$ to indicate expectation, probability, and mutual information taken with respect to a distribution $P$; however, when the distribution is clear from the context, the subscript will be omitted. We use a bold capital letter $\mathbf{P}$ to denote that a distribution $P$ is random. We use $\Rbb$ to denote the set of real numbers and $\Rbb^+$ to denote the nonnegative subset. 

For a distortion measure $d: \mathcal{X} \times \mathcal{Y}\mapsto \mathbb{R}^+$, we use $\Ebb \left[d(X,Y)\right]$ to measure the distortion of $X$ incurred by reconstructing it as $Y$. The maximum distortion is defined as
$$d_{max}=\max_{(x,y) \in \Xcal\times\Ycal} d(x,y).$$
The distortion between two sequences is defined to be the per-letter average distortion 
$$d(x^n,y^n)=\frac1n\sum_{t=1}^n d(x_t,y_t).$$

\subsection{Total Variation Distance}
The total variation distance between two probability measures $P$ and $Q$ on the same $\sigma$-algebra $\mathcal{F}$ of subsets of the sample space $\Xcal$ is defined as
$$\lVert P-Q\rVert_{TV}\triangleq \sup_{\mathcal{A}\in \mathcal{F}}|P(\mathcal{A})-Q(\mathcal{A})|.$$
\begin{property}[Property 2 \cite{schieler-journal}] \label{property-tv}
The total variation distance satisfies the following properties:
\begin{enumerate}[(a)]
%\item If $\Xcal$ is countable, then total variation can be rewritten as 
%\begin{equation}
%\lVert P - Q \rVert_{TV} = \frac12 \sum_{x\in\Xcal} |P(x)-Q(x)|.
%\end{equation}
\item \label{a} Let $\eps>0$ and let $f(x)$ be a function in a bounded range with width $b \in\Rbb$. Then
\begin{equation}
\label{tvcontinuous}
\lVert P-Q \rVert_{TV} < \eps \:\Longrightarrow\: \big| \Ebb_P[f(X)] - \Ebb_Q[f(X)] \big | < \eps b.
\end{equation}
\item \label{b} Total variation satisfies the triangle inequality. For any $S \in \Delta_{\Xcal}$, 
\begin{equation}
\lVert P - Q \rVert_{TV} \leq \lVert P - S \rVert_{TV} + \lVert S - Q \rVert_{TV}.
\end{equation}
\item \label{c} Let $P_{X}P_{Y|X}$ and $Q_XP_{Y|X}$ be two joint distributions on $\Delta_{\Xcal\times\Ycal}$. Then 
\begin{equation}
\lVert P_XP_{Y|X} - Q_X P_{Y|X} \rVert_{TV} = \lVert P_X - Q_X \rVert_{TV}.
\end{equation}
\item \label{d} For any $P,Q \in \Delta_{\Xcal\times\Ycal}$, 
\begin{equation}
\lVert P_X - Q_X \rVert_{TV} \leq \lVert P_{XY} - Q_{XY} \rVert_{TV}.
\end{equation}
\end{enumerate}
\end{property}

\subsection{Soft-covering Lemmas}
We now introduce two versions of soft-covering lemma, which will be used for the achievability proof. The basic soft-covering lemma has been used to obtain lossy source coding results \cite{cuff-itw2013} and \cite{song-isit2014}. However, a generalized superposition soft-covering lemma is required for meeting secrecy constraints.

\begin{lem}(\textbf{Basic soft-covering}, \cite{cuff2012distributed}] \label{bsc}
Given a joint distribution $P_{XY}$, let $\Ccal^{(n)}$ be a random collection of sequences $Y^n(m)$, with $m=1,...,2^{nR}$, each drawn independently and i.i.d. according to $P_Y$. Denote by $\Pbf_{X^n}$ the output distribution induced by selecting an index $m$ uniformly at random and applying $Y^n(m)$ to the memoryless channel specified by $P_{X|Y}$. Then if $R>I(X;Y)$,
$$\Ebb_{\cn}\left[\lVert \Pbf_{X^n}-\prod_{t=1}^n P_X\rVert_{TV}\right]\leq \epsilon_n,$$
where $\epsilon_n\rightarrow_n 0$.
\end{lem}

\begin{lem}(\textbf{Generalized superposition soft-covering}, \cite{schieler-journal}) \label{gsc}
Given a joint distribution $P_{UVXZ}$, let $\Ccal_U^{(n)}$ be a random codebook of $2^{nR_1}$ sequences in $\Ucal^n$, each drawn independently according to $\prod_{t=1}^n P_U(u_t)$ and indexed by $m_1\in[1:2^{nR_1}]$. For each $m_1$, let $\Ccal_V^{(n)}(m_1)$ be a random codebook of $2^{nR_2}$ sequences in $\Vcal^n$, each drawn independently according to $\prod_{t=1}^nP_{V|U}(v_i|u_i(m_1))$ and indexed by $(m_1,m_2)\in[1:2^{nR_2}]$.
Let 
\begin{eqnarray}
&&\Pbf_{M_1M_2X^nZ^k}(m_1,m_2,x^n,z^k)\nonumber\\
&\triangleq& 2^{n(R_1+R_2)}\prod_{t=1}^n P_{X|UV}(x_t|U_t(m_1),V_t(m_1,m_2))\nonumber\\
&& P_{Z|XUV}(z_t|x_t,u_t,v_t)^{1\{t\in[k]\}},
\end{eqnarray}
and
\begin{eqnarray}
&&\Qbf_{M_1X^nZ^k}(m_1,x^n,z^k)\nonumber\\
&\triangleq&2^{-nR_1}\prod_{t=1}^n P_{X|U}(x_i|U_i(m_1))\nonumber\\
&&P_{Z|XU}(z_i|x_i,U_i(m_1))^{1\{t\in[k]\}}
\end{eqnarray}

If $R_2>I(X;V|U)$, then there exists $\alpha\in(0,1]$, depending only on the gap $R_2-I(X;V|U)$, such that if $k<\lfloor\alpha n\rfloor$, then
\begin{eqnarray}
\Ebb_{\Ccal^{(n)}}\left[\Vert \Pbf_{M_1X^nZ^k}-\Qbf_{M_1X^nZ^k}\Vert_{TV}\right]\leq e^{-\gamma n}\rightarrow_n 0
\end{eqnarray}
for some $\gamma>0$.
\end{lem}

\section{Problem Setup and Main Results}\label{result}
\subsection{Problem Setup}
We want to determine the rate-distortion region for a secrecy system with an i.i.d. source and two side information sequences $(X^n, B^n, W^n)$ distributed according to $\prod_{t=1}^n \Pbar_{XBW}(x_t,b_t,w_t)$ satisfying the following constraints:
\begin{itemize}
\item Encoder $f_n: \mathcal{X}^n \mapsto \mathcal{M}$ (possibly stochastic);
\item Legitimate receiver decoder $g_n: \mathcal{M}\times \mathcal{B}^n \mapsto {\mathcal{Y}}^n$ (possibly stochastic);
\item Eavesdropper decoder $P_{Z^n|MW^n}$;
\item Compression rate: $R$, i.e. $|\mathcal{M}|=2^{nR}$.
\end{itemize}
The system performance is measured according to the following distortion metrics:
\begin{itemize}
\item Average distortion for the legitimate receiver: 
$$\Ebb[d_b(X^n,Y^n)]\leq_n D_b$$
\item Minimum average distortion for the eavesdropper: 
$$\min_{P_{Z^n|MW^n}}\Ebb[d_w(X^n,Z^n)]\geq_n D_w$$
\end{itemize}
Note that $d_b$ and $d_w$ can be the same or different distortion measures.

\begin{defn}
The rate-distortion triple $(R,D_b,D_w)$ is achievable if there exists a sequence of rate $R$ encoders and decoders $(f_n,g_n)$ such that
$$\Ebb[d_b(X^n,Y^n)]\leq_n D_b$$
and
$$\min_{P_{Z^n|MW^n}}\Ebb[d_w(X^n,Z^n)]\geq_n D_w.$$
\end{defn}

The above mathematical formulation is illustrated in Fig.\ref{setup}. 
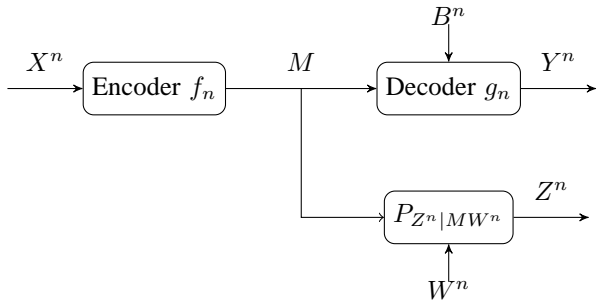
\begin{figure}
  \centering
\begin{tikzpicture}
[node distance=1cm,minimum height=7mm,minimum width=10mm,arw/.style={->,>=stealth'}]
  \node[coordinate] (source) {};
  \node[rectangle,draw,rounded corners] (encoder) [right =of source] {Encoder $f_n$};
  \node[rectangle,draw,rounded corners] (decoder) [right =2cm of encoder] {Decoder $g_n$};
  \node[coordinate] (sink)[right =of decoder] {};
  \node[coordinate] (side)[above =0.5cm of decoder] {};
  \node[rectangle,draw,rounded corners] (decoder2) [below =of decoder] {$P_{Z^n|MW^n}$}; 
  \node[coordinate] (mid)[right=1cm of encoder]{};
  \node[coordinate] (sink2)[right =of decoder2]{};
  \node[coordinate] (side2)[below =0.5cm of decoder2] {};
    
  \draw [arw] (source) to node[midway,above]{$X^n$} (encoder);
  \draw [arw] (encoder) to node[midway,above]{$M$} (decoder);  
  \draw [arw] (decoder) to node[midway,above]{$Y^n$} (sink);
  \draw [arw] (side) to node[midway,above]{$B^n$} (decoder);
  \draw [->] (mid) |- (decoder2);
  \draw [arw] (decoder2) to node[midway, above]{$Z^n$} (sink2);
  \draw [arw] (side2) to node[midway,below]{$W^n$} (decoder2);
\end{tikzpicture}
\caption{Secrecy system setup with side information at the decoders}
\label{setup}
\end{figure}

For the special case of lossless compression between the transmitter and the legitimate receiver, we make the following definition.
\begin{defn}
A rate-distortion pair $(R,D_w)$ is achievable if there exists a sequence of encoders and decoders $(f_n, g_n)$ such that 
$$\lim_{n\rightarrow \infty}\mathbb{P}\left[X^n\neq Y^n\right] = 0$$
and
$$ \min_{P_{Z^n|M,W^n}} \mathbb{E} [d_w(X^n, Z^n)]\geq_n D_w.$$
\end{defn}

\subsection{Less Noisy and More Capable Side Information}
\begin{defn}
The side information $B$ is strictly less noisy than the side information $W$ with respect to $X$ if
$$I(V;B)> I(V;W)$$
for all $V$ such that $V\inout X\inout (B,W)$ and $I(V;B)>0$.
\end{defn}

\begin{defn}
The side information $B$ is strictly more capable than the side information $W$ with respect to $X$ if
$$I(X;B)> I(X;W).$$
\end{defn}

\subsection{Main Achievability Result}
\begin{thm} \label{inner}
A rate-distortion triple $(R,D_b,D_w)$ is achievable if 
\begin{eqnarray}
&&R> I(V;X|B)\\
&&D_b\geq\Ebb [d_b(X,Y)]\\
&&D_w\leq\min_{z(u,w)}\Ebb [d_w(X,Z(U,W))]\\
&&I (V;B|U)>I (V;W|U)
\end{eqnarray}
for some $\Pbar_{UVXBW}=\Pbar_{XBW}\Pbar_{V|X}\Pbar_{U|V}$, where $Y=\phi(V,B)$ for some function $\phi(\cdot,\cdot)$.
\end{thm}
The proof of the above theorem is provided in the Section \ref{achievability}.

Theorem \ref{inner} involves two auxiliary variables $U$ and $V$ that are correlated with the source $X$ in a Markov chain relationship.  The variable $V$ can be understood as the lossy representation of $X$ that is communicated efficiently using random binning to the intended receiver, which will be used with the side information $B$ to estimate $X$, just as in the setting without an eavesdropper which was pioneered by \cite{wyner-ziv}.  The purpose of the auxiliary variable $U$ is to provide secrecy similar to the way secrecy is achieved in \cite{villard-allerton}.  The side information at the intended receiver must be better than that of the eavesdropper (as measured by mutual information with $V$) in order to prevent decoding of $V$.  The variable $U$ (if needed) is given away to all parties as the first layer of a superposition code in order to generate this condition for $V$.

\subsection{A Trivial Converse}
A tight outer bound is not attained and hence, the optimality of Theorem \ref{inner} is not yet known. A trivial outer bound is stated as follows for completeness.
\begin{thm} \label{outer}
If a rate-distortion triple $(R,D_b,D_w)$ is achievable, then 
\begin{eqnarray}
&&R> I(V;X|B) \label{rate-wz}\\
&&D_b\geq\Ebb [d_b(X,Y)] \label{distortion-bob}\\
&&D_w\leq \min_{z(w)} \Ebb [d_w(X,Z(W))] \label{distortion-eve}
\end{eqnarray}
for some $\Pbar_{VXBW}=\Pbar_{XBW}\Pbar_{V|X}$, where $Y=\phi(V,B)$ for some function $\phi(\cdot,\cdot)$ and all the quantities are with respect to $\Pbar_{XBW}$.
\end{thm}
\begin{proof}
To get $(\ref{rate-wz})$ and $(\ref{distortion-bob})$, we just need to apply the Wyner-Ziv converse; and to get $(\ref{distortion-eve})$, observe that the reconstruction cannot be worse than the symbol-by-symbol estimation of $X^n$ from $W^n$ without using $M$.
\end{proof}

\subsection{Less Noisy Side Information}
\begin{cor} \label{less-noisy}
If the legitimate receiver has \textbf{strictly} less noisy side information than the eavesdropper, the converse of Theorem \ref{outer} is tight.
%then the rate-distortion triple $(R,D_b,D_w)$ is achievable if and only if
%\begin{eqnarray}
%&&R\geq I(V;X|B)\label{con-r}\\
%&&D_b \geq \Ebb \left[d_b(X,Y)\right]\label{con-db}\\
%&&D_w \leq  \min_{z(w)} \Ebb[d_w(X,Z(W))]\label{con-dw}
%\end{eqnarray}
%for some $\Pbar_{VXBW}=\Pbar_{XBW}\Pbar_{V|X}$, where $Y=\phi(V,B)$ for some function $\phi(\cdot, \cdot)$.
\end{cor}
\begin{proof}
To see the achievability, we just need to set the $U$ in Theorem \ref{inner} to be $\varnothing$.
\end{proof}

Note that the strictly less noisy condition meets the inequality in Theorem \ref{inner}. Corollary \ref{less-noisy} covers the case of degraded side information at the eavesdropper, i.e. $X-B-W$, except for the corner case where $I(X;W)=I(X;B)$.

\subsection{Lossless Compression}
When the legitimate receiver must reconstruct the source sequence losslessly, we have the following inner bound.
\begin{cor} \label{lossless-ach}
$(R,D_w)$ is achievable if 
\begin{eqnarray}
&&R> H(X|B)\\
&&D_w\leq \min_{z(u,w)}\mathbb{E}[d_w(X,z(U,W))]\\
&&I(X;B|U)>I(X;W|U)  \label{ineq}
\end{eqnarray}
for some $\Pbar_{UXBW}=\Pbar_{XBW}\Pbar_{U|X}$. 
\end{cor}
\begin{proof}
This is consistent with Theorem \ref{inner} by setting $V=X$ and that the additional proof required for lossless recovery follows naturally from the construction of the achievability scheme for Theorem \ref{inner}.
\end{proof}

\begin{cor}\label{lossless-capable}
If the legitimate receiver has strictly more capable side information than the eavesdropper with respect to the source, then the rate-distortion pair $(R,D_w)$ is achievable if and only if
\begin{eqnarray}
&&R\geq H(X|B)\\
&&D_w\leq \min_{z(w)}\mathbb{E}[d_w(X,z(W))].
\end{eqnarray}
\end{cor}

\section{Proof of Achievability}\label{achievability}
We now give the achievability proof of Theorem \ref{inner} using the soft-covering lemmas. We apply the same proof technique using the likelihood encoder as introduced in \cite{song-isit2014} with the modification of using a superposition codebook.

The source is encoded into four messages $M_p$, $M_p'$, $M_s$ and $M_s'$, where $M_p$ and $M_s$ are transmitted and $M_p'$ and $M_s'$ are virtual messages that are not physically transmitted, but will be recovered with small error at the legitimate receiver with the help of the side information. On the other hand, $M_p$ and $M_p'$ play the role of public messages, which both the legitimate receiver and the eavesdropper will decode; $M_s$ and $M_s'$ index a codeword that is kept secret from the eavesdropper, which only the legitimate receiver can make sense of with its own side information.

Fix a distribution $\Pbar_{UVXBW}=\Pbar_U\Pbar_{V|U}\Pbar_{X|V}\Pbar_{BW|X}$ satisfying
\begin{eqnarray*}
I_{\Pbar} (V;B|U)>I_{\Pbar} (V;W|U),\\
\Ebb_{\Pbar} [d_b (X,\phi(V,B))] \leq D_b,\\
\min_{z(u,w)}\Ebb_{\Pbar} [d_w(X,Z(U,W))] \geq D_w,
\end{eqnarray*}
and fix rates $R_p$, $R_p'$, $R_s$, $R_s'$ such that 
\begin{eqnarray}
R_p+R_p'>I_{\Pbar}(U;X),\nonumber\\
R_p'<I_{\Pbar}(U;B),\nonumber\\
R_s+R_s'>I_{\Pbar}(X;V|U), \nonumber\\
I_{\Pbar}(V;W|U)<R_s'<I_{\Pbar}(V;B|U).\nonumber
\end{eqnarray}
The distribution induced by the encoder and decoder is 
\small
\begin{eqnarray}
&&\Pbf(x^n,b^n,w^n,m_p,m_p',m_s,m_s',\hat{m}_p',\hat{m}_s',y^n)\nonumber\\
&\triangleq&  \Pbar_{X^nB^nW^n}(x^n,b^n,w^n) \Pbf_{E}(m_p,m_p',m_s,m_s'|x^n)\nonumber\\
&&\Pbf_D(\hat{m}_p',\hat{m}_s'|m_p,m_s,b^n) \Pbf_{\Phi}(y^n|m_p,\hat{m}_p',m_s,\hat{m}_s',b^n),
\end{eqnarray}
\normalsize
where $\Pbf_{E}(m_p,m_p',m_s,m_s'|x^n)$ is the source encoder; $\Pbf_D(\hat{m}_p',\hat{m}_s'|m_p,m_s,b^n)$ is the first part of the decoder that estimates $m_p'$ and $m_s'$ as $\hat{m}_p'$ and $\hat{m}_s'$; $\Pbf_{\Phi}(y^n|m_p,\hat{m}_p',m_s,\hat{m}_s',b^n)$ is the second part of the decoder that reconstructs the source sequence.

{\textbf{Codebook generation}}: We independently generate $2^{n(R_p+R_p')}$ sequences in $\Ucal^n$ according to $\prod_{t=1}^n\Pbar_{U}(u_t)$ and index by $(m_p,m_p')\in[1:2^{nR_p}]\times[1:2^{nR_p'}]$. We use $\Ccal_U^{(n)}$ to denote this random codebook. For each $(m_p,m_p')\in[1:2^{nR_p}]\times[1:2^{nR_p'}]$, we independently generate $2^{n(R_s+R_s')}$ sequences in $\Vcal^n$ according to $\prod_{t=1}^n\Pbar_{V|U}(v_t|u_t(m_p,m_p'))$ and index by $(m_p,m_p',m_s,m_s')$, $(m_s,m_s')\in [1:2^{nR_s}]\times[1:2^{nR_s'}]$. We use $\Ccal_V^{(n)}(m_p,m_p')$ to denote this random codebook.

{\textbf{Encoder}}: The encoder $\Pbf_{E}(m_p,m_p',m_s,m_s'|x^n)$ is a likelihood encoder \cite{song-isit2014} that chooses $M_p, M_p', M_s, M_s'$ stochastically according to the following probability:
$$\Pbf_{E}(m|x^n)=\frac{\Lcal(m|x^n)}{\sum_{\bar{m}\in \Mcal}\Lcal(\bar{m}|x^n)}$$
where $m=(m_p,m_p',m_s,m_s')$, $\Mcal=[1:2^{nR_p}]\times[1:2^{nR_p'}]\times[1:2^{nR_s}]\times[1:2^{nR_s'}]$, and 
$$\Lcal(m|x^n)=\Pbar_{X^n|V^n}(x^n|v^n(m)).$$

{\textbf{Decoder}}: The decoder has two parts. Let $\Pbf_D(\hat{m}_p',\hat{m}_s'|m_p,m_s,b^n)$ be a good channel decoder with respect to the superposition sub-codebook $\{v^n(m_p,a_p,m_s,a_s)\}_{a_p,a_s}$ and the memoryless channel $\Pbar_{B|V}$. For the second part of the decoder, fix a function $\phi(\cdot,\cdot)$. Define $\phi^n(v^n,b^n)$ as the concatenation $\{\phi(v_t,b_t)\}_{t=1}^n$ and set the decoder $\Pbf_{\Phi}$ to be the deterministic function
\begin{eqnarray}
&&\Pbf_{\Phi}(y^n|m_p,\hat{m}_p',m_s,\hat{m}_s',b^n)\nonumber\\
&\triangleq& {1}\{y^n=\phi^n(v^n(m_p,\hat{m}_p',m_s,\hat{m}_s'),b^n)\}.\nonumber
\end{eqnarray}

{\textbf{Analysis}}: We examine the distortions at the two receivers one at a time. To analyze the distortion at the legitimate receiver, we will consider four distributions, the induced distribution $\Pbf$, two approximating distributions $\Qbf^{(1)}$ and $\Qbf^{(2)}$, and an auxiliary distribution $\Qbf'$ that helps with the analysis. The idea is to show that 1) the system has nice behavior for distortion under $\Qbf^{(2)}$; and 2) $\Pbf$ and $\Qbf^{(2)}$ are close in total variation (on average over the random codebook) through $\Qbf^{(1)}$. To analyze the distortion at the eavesdropper, we will consider the induced distribution $\Pbf$ together with an auxiliary distribution $\tilde{\Qbf}$.

\subsection{Distortion at the Legitimate Receiver}
This part of the proof follows the same idea of the achievability proof for the Wyner-Ziv setting using the likelihood encoder given in \cite{song-isit2014}. For clarity, we outline the key steps and some technical details are referred to \cite{song-isit2014}.

The approximating distributions $\Qbf^{(1)}$ and $\Qbf^{(2)}$ are defined through an idealized distribution $\Qbf$ of the structure given in Fig.\ref{superposition}. This idealized distribution $\Qbf$ can be written as
\begin{eqnarray}
&&\Qbf(x^n,b^n,w^n,m_p,m_p',m_s,m_s',u^n,v^n)\nonumber\\
&=&Q(m_p,m_p',m_s,m_s')\Qbf(u^n|m_p,m_p')\Qbf(v^n|u^n,m_s,m_s')\nonumber\\
&&\Qbf(x^n,b^n,w^n|m_p,m_p',m_s,m_s')\\
&=&\frac{1}{2^{n(R_p+R_p'+R_s+R_s')}}{1}\{u^n=U^n(m_p,m_p')\}\nonumber\\
&&1\{v^n=V^n(m_p,m_p',m_s,m_s')\}\nonumber\\
&&\Pbar_{X^nB^nW^n|V^n}(x^n,b^n,w^n|V^n(m_p,m_p',m_s,m_s')) \\
&=&\frac{1}{2^{n(R_p+R_p'+R_s+R_s')}}{1}\{u^n=U^n(m_p,m_p')\}\nonumber\\
&&1\{v^n=V^n(m_p,m_p',m_s,m_s')\}\nonumber\\
&&\prod_{t=1}^n \Pbar_{X|V}(x_t|v_t)\Pbar_{BW|X}(b_t,w_t|x_t), \label{Qf}
\end{eqnarray}
where $(\ref{Qf})$ follows from the Markov relation $V\inout X\inout BW$.

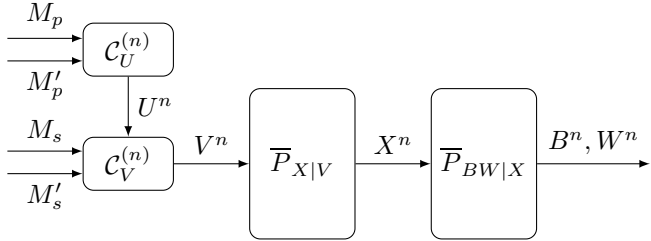
\begin{figure}[h]
\begin{center}
\begin{tikzpicture}[node distance=1cm]
 \node (src1)   [coordinate] {};
 \node (src11) [coordinate,above=0.15cm of src1] {};
 \node (src12) [coordinate, below=0.15cm of src1] {};
 \node (enc1)   [node,minimum width=12mm,right=10mm of src1] {$\Ccal_U^{(n)}$};
 \node (src2)   [coordinate,below=1.5cm of src1] {};
 \node (src21) [coordinate,above=0.15cm of src2] {};
 \node (src22) [coordinate, below=0.15cm of src2] {};
 \node (chv)    [coordinate,below=5mm of src1] {};
 \node (label1) [coordinate,right=2cm of src1] {};
 \node (enc2)   [node,minimum width=12mm,below=0.8cm of enc1] {$\Ccal_V^{(n)}$};
 \node (split1) [coordinate] at (src1 -| enc2.center) {};
 \node (chh)    [coordinate,right=27mm of enc2] {};
 \node (ch)     [node,minimum width=14mm,minimum height=2cm,right=1cm of enc2]  {$\Pbar_{X|V}$};
 \node (ch2)     [node,minimum width=14mm,minimum height=2cm,right=1cm of ch]  {$\Pbar_{BW|X}$};
 \node (sink) [coordinate, right=1.5cm of ch2] {};

 \draw[arw] (src11) to node[midway,above]{$M_p$} (src11 -| enc1.west);
 \draw[arw] (src12) to node[midway,below]{$M_p'$} (src12 -| enc1.west);
 \draw[arw] (src21) to node[midway,above]{$M_s$} (src21 -| enc2.west);
 \draw[arw] (src22) to node[midway,below]{$M_s'$} (src22 -| enc2.west);
\draw[arw] (enc1) to node[right]{$U^n$} (enc2.north);
 \draw[arw] (enc2) to node [midway,above] {$V^n$} (ch);
 \draw[arw] (ch) to node[midway,above]{$X^n$} (ch2);
 \draw[arw] (ch2) to node[midway,above]{$B^n,W^n$} (sink);
\end{tikzpicture}
\caption{Idealized distribution $\Qbf$ via a superposition codebook and memoryless channels $\Pbar_{X|V}$ and $\Pbar_{BW|X}$.}
\label{superposition}
\end{center}
\end{figure}

Note that the encoder $\Pbf_E$ satisfies
\begin{eqnarray}
\Pbf_E(m_p,m_p',m_s,m_s'|x^n)= \Qbf(m_p,m_p',m_s,m_s'|x^n) \label{PE}.
\end{eqnarray}

Furthermore, it can be verified with the same technique used in \cite{song-isit2014} that the idealized distribution $\Qbf$ satisfies:
\begin{eqnarray}
&&\Ebb_{\Ccal^{(n)}}\left[\Qbf(x^n,b^n,w^n,u^n,v^n)\right]\nonumber\\
&=&\Pbar_{X^nB^nW^nU^nV^n}(x^n,b^n,w^n,u^n,v^n), \label{Qexpectation}
\end{eqnarray}
where $\Ebb_{\Ccal^{(n)}}[\cdot]$ denotes $\Ebb_{\Ccal_U^{(n)}}\left[\Ebb_{\Ccal_V^{(n)}}\left[\cdot \right]\right].$

We now define the distributions $\Qbf^{(1)}$ and $\Qbf^{(2)}$ via the idealized distribution $\Qbf$ as follows:
\small
\begin{eqnarray}
&&\Qbf^{(1)}(x^n,b^n,w^n,u^n,v^n,m_p,m_p',m_s,m_s',\hat{m}_p',\hat{m}_s')\nonumber\\
&\triangleq&\Qbf(x^n,b^n,w^n,m_p,m_p',m_s,m_s',u^n,v^n)\nonumber\\
&&\Pbf_D(\hat{m}_p',\hat{m}_s'|m_p,m_s,b^n)\Pbf_{\Phi}(y^n|m_p,\hat{m}_p',m_s,\hat{m}_s')
\end{eqnarray}
\begin{eqnarray}
&&\Qbf^{(2)}(x^n,b^n,w^n,u^n,v^n,m_p,m_p',m_s,m_s',\hat{m}_p',\hat{m}_s')\nonumber\\
&\triangleq&\Qbf(x^n,b^n,w^n,m_p,m_p',m_s,m_s',u^n,v^n)\nonumber\\
&&\Pbf_D(\hat{m}_p',\hat{m}_s'|m_p,m_s,b^n)\Pbf_{\Phi}(y^n|m_p,{m}_p',m_s,{m}_s').
\end{eqnarray}
\normalsize
Notice that the distributions $\Qbf^{(1)}$ and $\Qbf^{(2)}$ differ only in $\Pbf_{\Phi}$. From $(\ref{Qexpectation})$, it can be shown that the distortion under distribution $\Qbf^{(2)}$ averaged over the random codebook is given by the following:
\small
\begin{eqnarray}
&&\Ebb_{\Ccal^{(n)}}\left[\Ebb_{\Qbf^{(2)}}[d_b(X^n,Y^n)]\right]\nonumber\\
&=&\sum_{x^n,v^n,b^n}\Ebb_{\Ccal^{(n)}}\left[\Qbf(x^n,v^n,b^n)\right]d_b(x^n,\phi^n(v^n,b^n))\\
&=&\sum_{x^n,v^n,b^n}\Pbar_{X^nV^nB^n}(x^n,v^n,b^n)d_b(x^n,\phi^n(v^n,b^n))\\
&=&\Ebb_{\Pbar}\left[d_b(X,Y)\right]. \label{dUnderQ2}
\end{eqnarray}
\normalsize

Define the auxiliary distribution $\Qbf'$ on a subset of the variables as
\small
\begin{eqnarray}
\Qbf'(m_p,m_p',x^n)\triangleq \frac{1}{2^{n(R_p+R_p')}}\Pbar_{X^n|U^n}(x^n|U^n(m_p,m_p')).
\end{eqnarray}
\normalsize

Since $R_s+R_s'>I_{\Pbar}(X;V|U)$, applying the generalized superposition soft-covering lemma, we have
\begin{eqnarray}
\Ebb_{\Ccal^{(n)}}\left[\left\Vert \Qbf_{M_pM_p'X^n}-\Qbf'_{M_pM_p'X^n}\right\Vert_{TV}\right]\leq e^{-\gamma_2 n} \triangleq{\epsilon_2}_n.\label{Q2Qp}
\end{eqnarray}

Also since $R_p+R_p'>I_{\Pbar}(U;X)$, applying the basic soft-covering lemma, we have
\begin{eqnarray}
\Ebb_{\Ccal^{(n)}}\left[\left\Vert\Pbar_{X^n}-\Qbf'_{X^n}\right\Vert_{TV}\right]\leq e^{-\gamma_1 n}\triangleq{\epsilon_1}_n.\label{Qp2Pbar}
\end{eqnarray}

Using Property \ref{property-tv}$(\ref{b})$, $(\ref{Qp2Pbar})$, and $(\ref{Q2Qp})$, we obtain
\begin{eqnarray}
\Ebb_{\Ccal^{(n)}}\left[\left\Vert\Qbf_{X^n}-\Pbar_{X^n}\right\Vert_{TV}\right]\leq {\epsilon_1}_n+{\epsilon_2}_n\triangleq {\epsilon_3}_n.
\end{eqnarray}

Therefore, by definitions of $\Pbf$ and $\Qbf^{(1)}$ and Property \ref{property-tv}$(\ref{c})$, we have
\begin{eqnarray}
\Ebb_{\Ccal^{(n)}}\left[\left\Vert \Pbf-\Qbf^{(1)} \right\Vert_{TV}\right]\leq {\epsilon_3}_n \label{P2Q1}
\end{eqnarray}
where the distributions are taken over $X^nB^nW^nM_pM_p'M_sM_s'\hat{M}_p'\hat{M}_s'Y^n$.

On the one hand, we need to apply the Wyner-Ziv technique to complete the distortion bound at the legitimate receiver. 
Since $R_p'<I_{\Pbar}(U;B)$ and $R_s'<I_{\Pbar}(V;B|U)$, the codebooks are randomly generated, and $M_p'$ and $M_s'$ are uniformly distributed under $\Qbf$, it is well known that the maximum likelihood decoder (as well as a variety of other decoders) will drive the error probability to zero as $n$ goes to infinity. This can be seen from Fig. \ref{superposition}, by identifying for fixed $M_p$ and $M_s$, that $M_p'$ and $M_s'$ are the messages to be transmitted over the memoryless channel $\Pbar_{B|V}$ with the superposition codebook. Specifically,
\begin{eqnarray}
\Ebb_{\Ccal^{(n)}}\left[\Pbb_{\Qbf^{(1)}}\left[(\hat{M}_p', \hat{M_s'})\neq(M_p',M_s')\right]\right]\leq \delta_n \rightarrow_n 0.
\end{eqnarray}

With Lemma 2 of \cite{song-isit2014}, it can be shown that 
\begin{eqnarray}
&&\Ebb_{\Ccal^{(n)}}\bigg[\bigg\Vert\Qbf_{X^nB^nW^nM_p\hat{M}_p'M_s\hat{M}_s'}^{(1)}\nonumber\\
&&\ \ \ \ \ \ \ \ \ -\Qbf_{X^nB^nW^nM_p{M}_p'M_s{M}_s'}^{(2)}\bigg\Vert_{TV}\bigg]\nonumber\\
&\leq&\Ebb_{\Ccal^{(n)}}\left[\Pbb_{\Qbf^{(1)}}\left[(\hat{M}_p', \hat{M_s'})\neq(M_p',M_s')\right]\right]\\
&\leq&\delta_n. \label{Q12Q2}
\end{eqnarray}

Hence, by $(\ref{dUnderQ2})$, $(\ref{P2Q1})$ and $(\ref{Q12Q2})$ and Property \ref{property-tv}$(\ref{a})$ and $(\ref{b})$, we obtain
\begin{eqnarray}
&&\Ebb_{\Ccal^{(n)}}\left[\Ebb_{\Pbf}[d_b(X^n,Y^n)]\right]\nonumber\\
&\leq& \Ebb_{\Pbar}[d_b(X,Y)]+{d_b}_{max}({\epsilon_3}_n+\delta_n)\\
&\leq&D_b+{d_b}_{max}({\epsilon_3}_n+\delta_n). \label{dis}
\end{eqnarray}
This completes the distortion analysis at the legitimate receiver.

\subsection{Distortion at the Eavesdropper}
To evaluate the enforced distortion at the eavesdropper with the best possible decoder, we will consider two distributions: the system induced distribution $\Pbf$ and an auxiliary distribution $\tilde{\Qbf}^{(i)}$ defined as
\small
\begin{eqnarray}
&&{\tilde{\Qbf}}^{(i)}(m_p,m_p',m_s,m_s',u^n,x,w^n)\nonumber\\
&\triangleq&\frac{1}{2^{n(R_p+R_p'+R_s+R_s')}}{1}\{u^n=U^n(m_p,m_p')\}\nonumber\\
&&\prod_{t=1}^n\Pbar_{W|U}(w_t|U_t(m_p,m_p')) \Pbar_{X|WU}(x|w_i,U_i(m_p,m_p')). \label{defQt}
\end{eqnarray}
\normalsize

Note that under ${\tilde{\Qbf}}^{(i)}$, we have the markov relation
\begin{eqnarray}
X\inout U_i(M_p,M_p')W_i\inout M_pM_p'M_sM_s'W^n. \label{Qtilde-markov}
\end{eqnarray}

The auxiliary distribution $\tilde{\Qbf}^{(i)}$ has the following property:
\begin{eqnarray}
&&\mathbb{E}_{\mathcal{C}_{U^n}}\left[\tilde{\Qbf}^{(i)}(u^n,w^n,x)\right]\nonumber\\
%&=&\Pbar_{U^n}(u^n)\prod_{t=1}^n\Pbar_{W|U}(w_t|u_t)\Pbar_{X|WU}(x_i|w_i,u_i)\nonumber\\
&=&\prod_{t=1}^n\Pbar_U(u_t)\Pbar_{W|U}(w_t|u_t)\Pbar_{X|WU}(x|w_i,u_i). \label{expected-iid}
\end{eqnarray}

Recall that under distribution $\mathbf{Q}$, for fixed $M_s=m_s$,
\small
\begin{eqnarray}
&&\mathbf{Q}(m_p,m_p',m_s',w^n,x_i|m_s)\nonumber\\
&=&\frac{1}{2^{n(R_p+R_p'+R_s')}}\Pbar_{W^n|V^n}(w^n|V^n(m_p,m_p',m_s,m_s'))\nonumber\\
&&\Pbar_{X|WVU}(x_i|w_i,V_i(m_p,m_p',m_s,m_s'),U_i(m_p,m_p'))
\end{eqnarray}
\normalsize

Since $R_s'>I_{\Pbar}(V;W|U)$, by applying the generalized superposition soft-covering lemma, we have for fixed $m_s$,
\small
\begin{eqnarray}
\Ebb_{\Ccal^{(n)}}\left[\left\Vert\tilde{\Qbf}^{(i)}_{M_pM_p'W^nX}-\Qbf_{M_pM_p'W^nX_i}\right\Vert_{TV}\right]\leq e^{-\gamma_4 n}\triangleq{\epsilon_4}_n .
\end{eqnarray}
\normalsize
Averaging over $M_s$, we have
\begin{eqnarray}
\Ebb_{\Ccal^{(n)}}\left[\left\Vert\tilde{\Qbf}^{(i)}_{M_pM_p'M_sW^nX}-\Qbf_{M_pM_p'M_sW^nX_i}\right\Vert_{TV}\right]\leq {\epsilon_4}_n, \label{Q2Qt}
\end{eqnarray}
and by Property \ref{property-tv}$(\ref{b})$, $(\ref{P2Q1})$ and $(\ref{Q2Qt})$,
\begin{eqnarray}
&&\Ebb_{\Ccal^{(n)}}\left[\left\Vert\tilde{\Qbf}^{(i)}_{M_pM_p'M_sW^nX}-\Pbf_{M_pM_p'M_sW^nX_i}\right\Vert_{TV}\right]\nonumber\\
&\leq& {\epsilon_3}_n+{\epsilon_4}_n\triangleq {\epsilon_5}_n. \label{e5}
\end{eqnarray}

Also note that, since $R_p+R_p'>0$, we can invoke Lemma \ref{gsc} by identifying 
$$(R_1,R_2,U,V,X,Z)\leftarrow (0,R_p+R_p',\varnothing, U, \varnothing ,U),$$
where the left side symbols represents the symbols from Lemma \ref{gsc}. This gives us
\begin{eqnarray}
\Ebb_{\Ccal^{(n)}}\left[\left\Vert\tilde \Qbf^{(i)}_{u_i(M_p,M_p')}-\Pbar_U\right\Vert_{TV}\right]\leq e^{-\gamma_6n}\triangleq{\epsilon_6}_n.\label{e6}
\end{eqnarray}

Combining $(\ref{dis})$, $(\ref{e5})$ and $(\ref{e6})$,  we get
\begin{eqnarray}
&&\Ebb_{\Ccal^{(n)}}\bigg[\sum_{i=1}^n\left\Vert \Pbf_{M_pM_p'M_sW^nX_i}-\tilde\Qbf^{(i)}_{M_pM_p'M_sW^nX}\right\Vert_{TV}\nonumber\\
&&+\sum_{i=1}^n\left\Vert \tilde\Qbf^{(i)}_{u_i(M_p,M_p')}-\Pbar_U\right\Vert_{TV}\nonumber\\
&&+\left\vert\Ebb_{\Pbf}[d_b(X^n,Y^n)]-D_b\right\vert\bigg]\nonumber\\
&\leq&n{\epsilon_5}_n+n{\epsilon_6}_n+{d_b}_{max}({\epsilon_3}_n+\delta_n)\\
&\leq&ne^{-n\min(\gamma_1,\gamma_2, \gamma_4,\gamma_6)}+{d_b}_{max}({\epsilon_3}_n+\delta_n)\\
&\triangleq&\epsilon_n\rightarrow_n 0.
\end{eqnarray}
\normalsize

Therefore, there exists a codebook under which
\begin{eqnarray}
\sum_{i=1}^n\left\Vert P_{M_pM_p'M_sW^nX_i}-\tilde Q^{(i)}_{M_pM_p'M_sW^nX}\right\Vert_{TV}\leq \epsilon_n,\\
\sum_{i=1}^n\left\Vert \tilde Q^{(i)}_{u_i(M_p,M_p')}-\Pbar_U\right\Vert_{TV}\leq\epsilon_n, \label{tot}
\end{eqnarray} 
and
\begin{eqnarray}
\Ebb_{P}[d_b(X^n,Y^n)]\leq D_b+\epsilon_n.
\end{eqnarray}
\normalsize

Finally, the distortion at the eavesdropper can be lower bounded by
\small
\begin{eqnarray}
&&\min_{z^n(m_p,m_s,w^n)}\mathbb{E}_P \left[d_w(X^n,z^n(M_p,M_s,W^n))\right]\nonumber\\
&\geq&\min_{z^n(m_p,m_p',m_s,w^n)}\mathbb{E}_{P}\left[d_w(X^n,z^n(M_p,M_p',M_s,W^n))\right]\\
&=&\frac1n \sum_{i=1}^n\min_{z_i(m_p,m_p',m_s,w^n)}\nonumber\\
&&\ \ \ \ \ \ \ \ \ \ \ \ \ \ \ \ \mathbb{E}_{P}\left[d_w(X_i,z_i(M_p,M_p',M_s,W
^n))\right]\\
&\geq&\frac1n \sum_{i=1}^n\min_{z_i(m_p,m_p',m_s,w^n)}\nonumber\\
&&\ \ \ \ \ \ \ \ \ \ \ \ \ \ \ \ \mathbb{E}_{{\tilde{Q}^{(i)}}}\left[d_w(X,z_i(M_p,M_p',M_s,W
^n))\right]\nonumber\\
&&-\epsilon_n {d_w}_{max}\\
&=&\frac1n \sum_{i=1}^n \min_{z(u,w)}\mathbb{E}_{{\tilde{Q}^{(i)}}} \left[d_w(X,z(u_i(M_p,M_p'),W_i))\right]\nonumber\\
&&-\epsilon_n {d_w}_{max} \label{using-markov}\\
&\geq&\frac1n\sum_{i=1}^n\min_{z(u,w)}\Ebb_{\Pbar}\left[d_w(X,z(U,W))\right]-2\epsilon_n {d_w}_{max} \label{last}
\end{eqnarray}
where $(\ref{using-markov})$ uses the markov relation under ${\tilde{Q}}^{(i)}$ given in $(\ref{Qtilde-markov})$, and $(\ref{last})$ uses $\left\Vert\tilde Q^{(i)}_{u_i(M_p,M_p')}-\Pbar_U\right\Vert_{TV}\leq \epsilon_n$ from $(\ref{tot})$ and the fact that 
$$\tilde Q^{(i)}_{W_iX|U_i}(w_i,x|u_i)=\Pbar_{W|U}(w_i|u_i) \Pbar_{X|WU}(x|w_i,u_i)$$
from $(\ref{defQt})$.
\normalsize

This completes the distortion analysis at the eavesdropper.

\section{Example}
We give an example for lossless compression case with Hamming distortion measure for the eavesdropper. The Hamming distortion measure is defined as 
\begin{displaymath}
   d(x,y) = \left\{
     \begin{array}{lr}
       0,  &x=y\\
       1,  &\text{otherwise.}
     \end{array}
   \right.
\end{displaymath} 
Let $X^n$ be a sequence of i.i.d. $Bern(p)$ source, and let $B^n$ and $W^n$ be side information obtained through a binary erasure channel (BEC) and binary symmetric channel (BSC), respectively, i.e.
\begin{eqnarray*}
\Pbar_{X}(0)=1-\Pbar_{X}(1)=1-p,\\
\Pbar_{B|X}(e|x)=\alpha,\\
\Pbar_{W|X}(1-x|x)=\beta.
\end{eqnarray*}

This is illustrated in Fig. \ref{bsc-bec}. This type of side information was also considered in \cite{villard-journal}, but only with $Bern(\frac12)$ source.
%\begin{figure}
%  \centering
%\begin{minipage}{.2\textwidth} 
%\begin{tikzpicture}
%[node distance=1cm,minimum height=10mm,minimum width=14mm,arw/.style={->,>=stealth'}]
%  \node[coordinate] (source0) {} node[left] {$0$};
%  \node[coordinate] (source1) [below=1cm of source0] {}; 
%  \node[coordinate] (B0) [right=2cm of source0] {};
%  \node[coordinate] (Be) [below=0.5cm of B0] {};
%  \node[coordinate] (B1) [below=0.5cm of Be] {};
%
%  \draw[arw] (source0) to node[midway,above]{} (B0) node[right] {$0$};
%  \draw[arw] (source0) to node[midway,below]{} (Be) node[right] {$e$};  
%  \draw[arw] (source1) to node[midway,above]{} (B1) node[right] {$1$};
%  \draw[arw] (source1) to node[midway,above]{} (Be); 
%\end{tikzpicture}
%\end{minipage}
%\begin{minipage}{.2\textwidth} 
%\begin{tikzpicture}
%[node distance=1cm,minimum height=10mm,minimum width=14mm,arw/.style={->,>=stealth'}]
%  \node[coordinate] (source0) {};
%  \node[coordinate] (source1) [below=1cm of source0] {}; 
%  \node[coordinate] (B0) [right=2cm of source0] {};
%  \node[coordinate] (B1) [below=0.5cm of Be] {};
%
%  \draw[arw] (source0) to node[midway,above]{} (B0);
%  \draw[arw] (source0) to node[midway,above]{} (B1);  
%  \draw[arw] (source1) to node[midway,above]{} (B1);
%  \draw[arw] (source1) to node[midway,above]{} (B0); 
%\end{tikzpicture}
%\end{minipage}
%\caption{Side information $B^n$ and $W^n$}
%\label{bsc-bec}
%\end{figure}
\begin{figure}
\begin{center}
\begin{minipage}{.2\textwidth} 
\begin{tikzpicture}
%\tikzstyle{every node} = [circle, fill = gray!30]
\node (a) [circle] at (0,0) {$1$};
\node (b) [circle] at (0,2) {$0$};
\node (c) [circle] at (2,0) {$1$};
\node (d) [circle] at (2,2) {$0$};
\node (e) [circle] at (2,1) {$e$};
\node (X) [circle] at (0,2.7) {$X$};
\node (B) [circle] at (2,2.7) {$B$};
\draw[->] (a) -- (c) node[pos=.5,sloped,below] {$1-\alpha$};
\draw[->] (a) -- (e) node[pos=.45,below] {$\alpha$};
\draw[->] (b) -- (e) node[pos=.45,above] {$\alpha$};
\draw[->] (b) -- (d) node[pos=.5,sloped,above] {$1-\alpha$};
\end{tikzpicture}
\end{minipage}
\begin{minipage}{.2\textwidth} 
\begin{tikzpicture}
%\tikzstyle{every node} = [circle, fill = gray!30]
\node (a) [circle] at (0,0) {$1$};
\node (b) [circle] at (0,2) {$0$};
\node (c) [circle] at (2,0) {$1$};
\node (d) [circle] at (2,2) {$0$};
\node (X) [circle] at (0,2.7) {$X$};
\node (B) [circle] at (2,2.7) {$W$};
\draw[->] (a) -- (c) node[pos=.5,sloped,below] {$1-\beta$};
\draw[->] (a) -- (d) node[pos=.45,below] {$\beta$};
\draw[->] (b) -- (c) node[pos=.45,above] {$\beta$};
\draw[->] (b) -- (d) node[pos=.5,sloped,above] {$1-\beta$};
\end{tikzpicture}
\end{minipage}
\end{center}
\caption{Side information $B$ and $W$ correlated with source $X$}
\label{bsc-bec}
\end{figure}
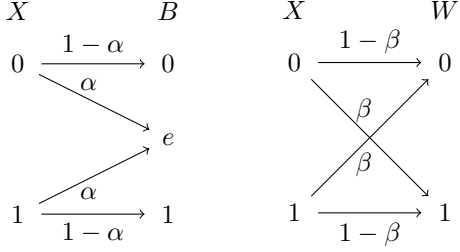

We consider a generic discrete auxiliary random variable $U$ that takes values on $1,...,|\Ucal|$ with $\Pbar_{U}(i)=u_i$ and $\Pbar_{X|U}(0|i)=\delta_i$, $\Pbar_{X|U}(1|i)=1-\delta_i$. It can be shown that the distortion $D_w$ takes the following form. By applying Corollary \ref{lossless-ach}, we can obtain the following theorem.
\begin{thm}
$(R,D_w)$ is achievable for the BEC-BSC side information with Hamming distortion $d_w(\cdot,\cdot)$ if 
\begin{eqnarray*}
R&\geq& \alpha h(p)\\
D_w&\leq& \max_{\{u_i,\delta_i\}_{i=1}^3} \sum_{i=1}^{3} u_i \min(\delta_i,1-\delta_i,\beta)\\
&s.t.&  0\leq u_i, \delta_i \leq 1\\
&&\sum_{i=1}^{3}u_i=1\\
&&\sum_{i=1}^{3}u_i\delta_i=1-p\\
&&\sum_{i=1}^{3}u_i[(1-\alpha)h(\delta_i)-h(\delta_i\ast\beta)]+h(\beta)\geq0
\end{eqnarray*}
where $h(\cdot)$ denotes the binary entropy function.
\end{thm}

We plot the distortion at the eavesdropper as a function of the source distribution $p$ for fixed $\alpha$ and $\beta$ in Fig. \ref{plot1} and Fig. \ref{plot2}, where the outer bounds are calculated from Theorem \ref{outer}.
\begin{figure}[htbp]
  \centering
  \includegraphics[width=9 cm]{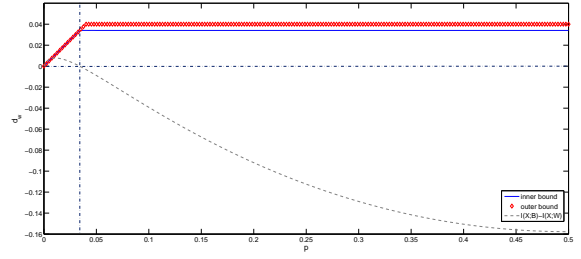}
\caption{Distortion at the eavesdropper as a function of source distribution $p$ with $\alpha=0.4$, $\beta=0.04$}
\label{plot1}
\end{figure}

\begin{figure}[htbp]
  \centering
  \includegraphics[width=9 cm]{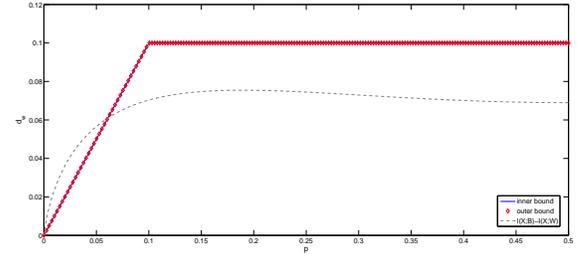}
\caption{Distortion at the eavesdropper as a function of source distribution $p$ with $\alpha=0.4$, $\beta=0.1$}
\label{plot2}
\end{figure}

In Fig. \ref{plot1}, when the legitimate receiver's side information is more capable than the eavesdropper's side information with respect to the source, perfect secrecy at the eavesdropper is achieved; when the eavesdropper's side information is more capable than the legitimate receiver, with our encoding scheme, we achieve a positive distortion at the eavesdropper with no additional cost on the compression rate to ensure lossless decoding at the legitimate receiver. It is worth noting that our scheme encodes the source so that it favors the side information for the legitimate receiver even if the legitimate receiver's side information is less capable, as opposed to the case where the regular Wyner-Ziv (Slepian-Wolf) encoding scheme that gives the same compression rate but no distortion at the eavesdropper. 

In Fig. \ref{plot2}, since the legitimate receiver's side information is always more capable than the eavesdropper's side information, it is a direct application of Corollary \ref{lossless-capable} and perfect secrecy is ensured.

\section{Conclusion}
We have investigated the performance of a secrecy system with side information at receivers under the rate-distortion criteria. Our results show that even if the legitimate receiver has a weaker side information, a positive distortion can be enforced to the eavesdropper. Although exact bounds have been obtained for several special cases, the outer bound for arbitrarily correlated side information is not tight. This suggests an interesting direction for future work.
\bibliographystyle{ieeetr}

\section{Acknowledgement}
This research was supported in part by the Air Force Office of Scientific Research under Grant FA9550-12-1-0196 and MURI Grant FA9550-09-05086 and in part by National Science Foundation under Grants CCF-1116013, CNS-09-05086 and CCF-1350595.

\bibliography{lossy-sideinfo}
\end{document}